\documentclass{elsarticle}
\usepackage{amsmath,amssymb,url,color}
\usepackage{enumerate}
\newcommand{\EQ}{\begin{equation}}
\newcommand{\EN}{\end{equation}}
\newtheorem{theorem}{Theorem}
\newtheorem{definition}{Definition}
\newtheorem{corollary}{Corollary}
\newtheorem{proposition}{Proposition}
\newtheorem{lemma}{Lemma}

\newtheorem{example}{Example}
\newtheorem{remark}{Remark}

\newcommand{\FF}{{\mathbb F}}
\newcommand{\Fk}{{\mathbb F}_{2^k}}  
\newcommand{\F}{{\mathbb F}_{2^n}}

\newcommand{\G}{{\cal G}}

\newcommand{\J}{\mathcal{J}}

\newcommand{\W}{{\cal W}}
\newcommand{\Sp}{{\cal S}}

\newenvironment{proof}{\begin{trivlist}\item[]{\em Proof. }}%
{\samepage\hfill$\diamond$\end{trivlist}}
\date{}

\begin{document}

\title{Permutations via linear translators}
%\author[Cepak, Charpin and Pasalic]{Nastja Cepak, Pascale Charpin and Enes Pasalic}
\author[upr]{Nastja Cepak}
\ead{nastja.cepak@gmail.com}
\author[inria]{Pascale Charpin}
\ead{pascale.charpin@inria.fr}
\author[upr]{Enes Pasalic}
\ead{enes.pasalic6@gmail.com}

\address[upr]{University of Primorska, FAMNIT \& IAM, Glagolja\v ska 6, 6000 Koper, Slovenia}
\address[inria]{INRIA, 2 rue Simone Iff, Paris, France}
%\email{enes.pasalic6@gmail.com} 

\begin{abstract}
We show that many infinite  classes of permutations over finite fields can be constructed
via translators with a large choice of parameters.  We first characterize some functions
having linear translators, based on which  several families of permutations are then derived. 
Extending the results of \cite{kyu}, we give in several cases the
compositional inverse of these permutations. The connection 
with complete permutations is also utilized to provide further infinite classes of permutations. Moreover, we
propose new tools to study permutations of the form 
$x\mapsto x+(x^{p^m}-x+\delta)^s$ and a few infinite classes of permutations of this form are  proposed.
\end{abstract}

\maketitle
%\vspace{1cm}
%{\bf\Large Main draft}

\section{Introduction}
The main goal of this paper is to contribute to  the study of permutations of
finite fields. A finite field of order $p^n$ is denoted
$\FF_{p^n}$ where $p$ is any prime and $n$ a positive integer.
A polynomial $F\in\FF_{p^n}[x]$ is said to be a permutation if its associated
mapping $x\mapsto F(x)$ over $\FF_{p^n}$ is bijective. During the last few years there has been a tremendous progress in construction methods and characterization of many infinite classes of permutations, see a  survey on recent works in \cite{Hou14} and the references therein. The use of permutations in 
applications such as coding is well-known and understood. The bijectivity is also  an important cryptographic criterion used in the design
of some block ciphers. For applicative purposes the use of sparse permutations,
{\it i.e.,} which can be expressed with few terms, is also an important
 property along with  the degree and the nonlinearity
 which are referred  to as  the standard cryptographic criteria. For this reason,
we are mainly interested in specifying  design methods of sparse permutations, having a few polynomial terms.

\medskip
This paper is based on the work of Kyureghyan \cite{kyu} where
permutations over $\FF_{p^{rk}}$ of kind
\EQ\label{eq:form}
F~:~x\mapsto L(x)+L(\gamma) h(f(x)), ~f:\FF_{p^{rk}}\rightarrow \FF_{p^{k}},~
h:\FF_{p^{k}}\rightarrow \FF_{p^{k}},
\EN
are studied. Here $\gamma\in\FF^*_{p^{rk}}$ is a so-called 
{\it $b$-linear translator}
of $f$ (cf. Definition~\ref{de:tr}) and $L$ a linear permutation. Note that this construction
is in a certain sense a generalization of the so-called 
{\it switching construction} \cite{ChaKyu-ffa,ChKySu}. Akbary, Ghioca and
 Wang unified the Kyureghyan's construction  for arbitrary subsets $S \subset \FF_{p^{n}}$ (not only subfields of $\FF_{p^{n}}$) 
 along with proposing a few other constructions in \cite{AkGhWa11}. This general criterion is now 
called AGW criterion \cite[Theorem 8.1.39]{MullenWang}. After these pioneering works a series of papers  
\cite{TuZeLiHe15,TuZeHu14,TuZeJi15,YuDi} (among others) treated the same topic of specifying new classes of permutation polynomials of the above form. For a nice survey of recent achievements related to this particular class of permutations  the reader is referred to \cite{Hou14}. Nevertheless, most of the recent contributions attempt to specify suitable functions $h,f$ and $L$ as in (\ref{eq:form}), or alternatively, for $F$ given by 
\begin{equation}\label{eq:specclass}
F~:~x\mapsto \gamma(f(x)+\delta)^s+L(x),~\delta\in\FF^*_{p^n},
\end{equation}
 to specify suitable degree $s$, $\delta \in \FF_{p^n}$, the function $f$ and also some particular field characteristic $p$, see for instance \cite{TuZeLiHe15} where three classes of permutations of the form (\ref{eq:specclass}) were specified for $p=3$.

 Our main purpose is to emphasize that the use of 
functions $f$ which have translators gives us the possibility  to construct many
infinite classes of permutations with a large choice of parameters. A suitable use of  this method allows us also to construct linear permutations and sparse permutations
  of high degree and to give their compositional inverses. 
 Moreover, a connection of this class of permutations to complete
 permutations is  considered and also more general results related
 to an explicit specification of permutations of the form (\ref{eq:specclass}) are given (for instance valid for any degree $s$ for suitable $f$ and $\delta$). 

\medskip
More specifically this paper is organized as follows. After preliminaries, Section \ref{se:trans}
is devoted to the existence of translators $\gamma$ for functions  $f$, where $f$ has a sparse polynomial representation.
 In Section \ref{se:inv}, we are interested in the compositional
inverses of permutations of type (\ref{eq:form}), similarly to, for instance,  the  work of Tuxanidy and Wang \cite{TuxaWang}.
Provided that $f$ has a $b$-translator $\gamma$, the function $g:u\mapsto u+bh(u)$ must permute $\FF_{p^{k}}$  to ensure the bijectivity  of $F$ \cite{kyu}. Nevertheless, when $b=0$ this holds for any $h$
leading to several families of permutations with their compositional
inverse. It is shown later  that when $b\ne 0$, by defining a class of involutions in odd characteristic
(Proposition \ref{pr:inv2}), we are still able to specify  the compositional  inverses in certain cases.   

Permutations $F$ of type (\ref{eq:form}) are closely related to so-called  complete 
mappings through the condition that $g$ must be a permutation. However, 
note that $h$ does not need to be bijective to apply Theorem~\ref{th:main},
and therefore $g$ is not necessarily a complete permutation. The connection to complete permutations, which we explain and illustrate in Section \ref{se:cplte}, is rather made to relate the number of recent works on this topic for the purpose of specifying new classes of permutations.

In Section \ref{se:spec}, a special class of functions given by (\ref{eq:specclass}),
%\begin{equation}\label{eq:specclass}
%F~:~x\mapsto \gamma(f(x)+\delta)^s+L(x)~~,~\delta\in\FF^*_{p^n},
%\end{equation}
which has been  studied in several papers 
(see \cite{TuZeLiHe15,TuZeJi15,YuDiWaPi} and references therein),
 is considered. 
We first show that Theorem \ref{th:main} applies to this class of permutations when $\delta$ and $f$ 
satisfy some simple conditions (Proposition \ref{pr:perm}), which essentially  gives us the possibility of specifying a family of infinite classes of permutations for any $s$. This is also the main difference to many previous works e.g. \cite{TuZeLiHe15,TuZeJi15,YuDiWaPi}, where some specific classes of permutations were identified only for certain exponents $s$. Moreover,
we specify the conditions  that $F$, as specified above, is a permutation  for both $p=2$ and $p$ odd.
In both cases, we have been  able to adapt  Theorem \ref{th:main} and to satisfy these conditions, thus
providing other infinite classes of permutations (Propositions 
\ref{pr:var1} and \ref{pr:var2}).
Actually,  our generalized framework  turns out to give  another (simpler) method  to prove the bijectivity
of some functions studied in  \cite{TuZeJi15,YuDi,TuZeLiHe15}. 

 On the other hand, it turns out that the results in Section~\ref{se:spec} can be derived from the results in \cite{AkGhWa11}, more precisely from Theorem 5.1 and Proposition 5.9 in \cite{AkGhWa11}. Nevertheless, our proof technique may have independent significance in the analysis of similar classes of permutations and more importantly our approach may potentially give an insight in the spectra of the component functions which has a great importance in cryptographic applications. 

\section{Preliminaries}
We recall  some definitions or results given in \cite{kyu}.
 Throughout this paper $p$ designates any prime. 
\begin{definition}\label{de:tr}
Let $n=rk$, $1\leq k\leq n$. Let $f$ be a function from $\FF_{p^n}$ to 
$\FF_{p^k}$, $\gamma\in\FF_{p^n}^*$
and $b$ fixed in $\FF_{p^k}$.
Then $\gamma$ is a $b$-{\it linear translator} for $f$ if
\[
f(x+u\gamma)-f(x)=ub,~~\mbox{for all $x\in\FF_{p^n}$ and for all $u\in\FF_{p^k}$}.
\]
In particular, when  $k=1$, $\gamma$ is usually  said to be a 
 $b$-{\it linear structure} of the  function $f$ (where $b\in\FF_p$),
that is
\[
f(x+\gamma)-f(x)=b~~\mbox{for all $x\in\FF_{p^n}$}.
\]
\end{definition}
We denote by $Tr(\cdot)$ the absolute trace on $\F$ and by $T^n_k(\cdot)$
the trace function from $\FF_{p^n}$ to  $\FF_{p^k}$, where $k$ divides $n$:
\[
T^n_k(\beta)=\beta+\beta^{p^k}+\dots+\beta^{p^{(n/k-1)k}}.
\]
We have also to recall  that a $\FF_{p^k}$-linear function 
on $\FF_{p^n}$ ($n=rk$) is of the type 
\[
L:\FF_{p^n}\rightarrow \FF_{p^n},~L(x)=\sum_{i=0}^{r-1}\lambda_ix^{p^{ki}}~,~
\lambda_i\in\FF_{p^n}.
\]
In the case when $n=2k$, it is easy to describe such linear permutations.
The next lemma is proved useful in the sequel.
\begin{lemma}\label{le:lin}
Let $n=2k$ and $L:\FF_{p^n}\rightarrow \FF_{p^n}$, $L(x)=ax+bx^{p^k}$, where
$a,b\in\FF^*_{p^n}$. Let $\G$ be the subgroup of  $\FF^*_{p^n}$ of order $p^k+1$.
Then we have:
\begin{description}
\item[(i)] $L$ is a permutation if and only if  $ab^{-1}\not\in\G$;
\item[(ii)]  $L$ is an involution if and only if $T^n_k(a)=0$ and 
$b^{p^k+1}=1-a^2$.
\end{description}
\end{lemma}
\begin{proof}
Since $L(x)=x(a+bx^{p^k-1})$, $ab^{-1}\not\in\G$ means that the kernel of
$L$ is $\{0\}$.  Now we have
\[
L\circ L(x)=x(a^2+b^{p^k+1})+x^{p^k}b(a+a^{p^k}).
\]
Thus $L$ is an involution if and only if $a+a^{p^k}=0$ and $a^2+b^{p^k+1}=1$.
When $p$ is odd, note that $a+a^{p^k}=0$ implies $a^2\in \FF^*_{p^k}$.
The case $p=2$ is an instance of \cite[Proposition 5]{CMS}.
\end{proof}
The following general theorem is given in \cite{kyu}
without proof since the proof is an equivalent of those given in 
 \cite{ChaKyu-ffa} and \cite{ChSa11},  when $k=1$ and $k=n$, respectively.
\begin{theorem}
A   function $f$ from $\FF_{p^n}$ to 
$\FF_{p^k}$, $n=rk$, has a linear translator if and only if there is a
 non-bijective $\FF_{p^k}$-linear function $L$ on $\FF_{p^n}$ such that
\[
f(x)=T^n_k\left(H\circ L(x)+\beta x\right)
\]
for some $H:\FF_{p^n}\rightarrow \FF_{p^n}$ and $\beta\in\FF_{p^n}$. In this case
the kernel of $L$ is contained in the subspace of linear translators
(including $0$ by convention). 
\end{theorem}

Now we have the following construction, introduced by  Kyureghyan
in \cite[Theorem 1]{kyu}. This  result can also be obtained
by using the AGW criterion, see  Section 6 in \cite{AkGhWa11}.
\begin{theorem}{\rm\cite[Theorem 1]{kyu}}\label{th:main}
Let $n=rk$, with $r, k>1$. Let $L$ be a $\FF_{p^k}$-linear permutation on $\FF_{p^n}$.
Let $f$ a function from $\FF_{p^n}$ onto $\FF_{p^k}$, 
$h:\FF_{p^k}\rightarrow\FF_{p^k}$, $\gamma\in\FF_{p^n}^*$
and $b$ is fixed in $\FF_{p^k}$.
Assume that $\gamma$ is a $b$-linear translator of $f$. Then
$$F(x)=L(x)+L(\gamma)h(f(x))$$ permutes $\FF_{p^n}$ if and only if 
$g:u\mapsto u+bh(u)$ permutes $\FF_{p^k}$.
\end{theorem}

\section{On functions having translators}\label{se:trans}
In this section, motivated by the possibility of specifying new classes of permutations by means of Theorem~\ref{th:main}, we investigate the existence of linear translators for sparse polynomials $f:\FF_{p^n} \rightarrow \FF_{p^k}$ (the problem being difficult for arbitrary polynomials). More precisely, we  show the non-existence of linear translators for monomials and derive the exact form of binomials for which 
there exist linear translators. The monomial trace function of the form $Tr_k^n(x^d)$ is also considered.

The following two results are frequently used throughout this section. 

\begin{theorem}{\rm [Lucas' theorem]}
Let $a,b$ be  positive integers and $a=\sum_{i=1}^n a_ip^i$, $b=\sum _{i=1}^n b_ip^i$ their  $p$-adic expansions, where $a_i, b_i \in \FF_p$.
	Then $$ \binom{a}{b} \pmod{p}\equiv  \binom{a_1}{b_1}\cdots \binom{a_n}{b_n}.$$
	It follows that $ \binom{a}{b} \pmod p \neq 0$ if and only if 
$b\preceq a$, i.e., $b_i \leq a_i$ for all $i$.
\end{theorem}

Let now $f(x):\FF_{p^n}\rightarrow \FF_{p^n}$, $f(x)=\sum _{i=0}^{p^n-1}b_ix^i$.
In \cite{EnesDeriv},  a compact formula relating the coefficients $b_i$
 of $f$ and of its derivative 
$f(x+u\gamma )-f(x)=\sum_{t=0}^{p^n-2}c_t x^t$ was derived.  More precisely
\begin{equation}\label{eq:connect} 
c_t=\sum_{i=t+1}^{p^n-1}\binom{i}{t} (u\gamma )^{i-t}b_i,~~
t\in \lbrace 0, 1, \ldots , p^n-2 \rbrace.
\end{equation}
The first application of these results regards the existence of translators for $f:\FF_{p^n} \rightarrow \FF_{p^k}$ which is either monomial or binomial.
\begin{proposition}\label{le:dPk}
	Let $f(x)=x^d$, $f: \FF_{p^n} \rightarrow \FF_{p^k}$,
 where $n=rk$ and $r>1$.
	
	\begin{enumerate}[i)] 
	\item Then the image set of $f$ is in $\FF_{p^k}$ if and only if the exponent $d$ is of the form
	\begin{equation} \label{eq:dform} 
	d=j(p^{k(r-1)} + p^{k(r-2)} + \cdots +p^k+1),
\end{equation}	
	 for some $j\in \lbrace 1, \ldots , p^k-1\rbrace$.
	 \item The function $f$ does not have  a linear translator in sense of Definition~\ref{de:tr}.
	 \end{enumerate}
\end{proposition}
\begin{proof}
	$i)$ Since $f$  maps to some subfield $\FF _{p^k}$, $(x^d)^{p^k}=x^d$
 must be true. This means
	$x^{d(p^k-1)}=1$ and therefore $d(p^k-1)\equiv 0 \pmod{p^n-1}$. It follows that 
	\begin{eqnarray*}
	d &=& j\frac{p^n-1}{p^k-1}=j(p^{k(r-1)} + p^{k(r-2)} + \cdots + 1),
	\end{eqnarray*}
	for some $j\in \lbrace 1, \ldots , p^k-1\rbrace$. 

\medskip
\noindent
$ii)$ If a function $f(x)=\sum _{i=0}^{p^n-1}b_ix^i$ has
 a linear translator, it must satisfy two necessary but not sufficient
 conditions:
	\begin{enumerate}
\item it must map to a subfield $\FF _{p^k}$ as requested by the definition, and 
\item its coefficients $b_i$ must satisfy 
$c_t=0$, for $t\in \lbrace 1, \ldots , p^n-2 \rbrace$ and $c_0 \neq 0$, where 
$c_t$ and $c_0$ are defined above by (\ref{eq:connect}).
	\end{enumerate}

The first condition implies that $d$ must be of the form (\ref{eq:dform}), 
	for  $j\in \lbrace 1, \ldots , p^k-1\rbrace$.
		Since $b_i=0$ for $i \neq d$, the second condition implies that
	$c_t=\binom{d}{t}(u\gamma )^{d-t}=0$, for all 
$t\in \lbrace 1, \ldots ,d-1 \rbrace$. This is satisfied only if 
 $\binom{d}{t} \equiv 0 \pmod{p}$ for all $t$.
	Using Lucas' theorem, the only possibility is $t\npreceq d$,
 for all $t$. But since our $d$ satisfies (\ref{eq:dform}),
for some $j\in \lbrace 1, \ldots , p^k-2\rbrace$,
	 this is
	impossible.	
\end{proof}

\begin{proposition}\label{prop:binom}
	Let $f(x)=\beta x^i+x^j$, $i<j$, where $f:\FF_{p^n} \rightarrow \FF_{p^k}$, $\beta \in \FF^*_{p^n}$ and $n=rk$, where $r >1$. Then the function $f$ has a linear translator if and only if $n$ is even,
	$k=\frac{n}{2}$, and furthermore $f(x)=T^n_k(x)$.
\end{proposition}
\begin{proof}
	Let $f(x)=\beta x^i + x^j,i<j,\beta \neq 0$.
	The function $f$  must satisfy the same two properties as in
 the proof of Proposition~\ref{le:dPk}. 
The second property, according to Definition \ref{de:tr} 
and (\ref{eq:connect}), implies that $c_t$ must satisfy 
\begin{equation}\label{eq:coef}
0=c_t= \left\lbrace \begin{array}{ll}
	0 & \mathrm{\ for\ }j\leq t\leq p^n-2 \\
	\binom{j}{t}(u\gamma )^{j-t} & \mathrm{\ for\ }i \leq t < j \\ 
	\binom{i}{t}(u\gamma)^{i-t}\beta + \binom{j}{t}(u\gamma )^{j-t} & \mathrm{\ for\ }0 < t< i
	\end{array} \right. .
\end{equation}
	
	Suppose $i$ and $j$ are both powers of $p$ so that $i=p^{i'}, j=p^{j'}$.
 Since   $t\npreceq j$ and $t\npreceq i$ for any $t$ in the above range,
 by Lucas' theorem  $c_t=0$ for all $t \neq 0$.

Assume now that $i$ and $j$ are not both powers of $p$ and
that (\ref{eq:coef}) holds. First, we must have
$t\not\prec j$ for   $i \leq t < j$ (to have $c_t=0$ for such $t$);
in particular $i\not\prec j$. Then, there exists  $t$, $0<t<i$, such that 
either $t\prec j$ or   $t\prec i$ for $t <i$. 
Since $c_t=0$   we have:
\begin{itemize}
\item if $t\prec i$ then $t\prec j$, because otherwise $\beta=0$, 
a contradiction;
\item if $t\prec j$ then $t\prec i$, since otherwise 
$c_t=\binom{j}{t}(u\gamma )^{j-t}\ne 0$;
\end{itemize}
 Thus,   $t\prec i$ if and only if
  $ t\prec j$, for all $t\in \lbrace 1, \ldots , i-1\rbrace$.
But, since $i\not\prec j$ there is $t'<i$ which satisfies $t'\prec i$, 
and $t'\not\prec j$, a contradiction.  
	
	Let us now analyze when $f(x)=\beta x^{p^{i'}} + x^{p^{j'}}$.
 Note that we want to have
\[
f(x+\gamma u)-f(x)=f(\gamma u)=\beta (\gamma u)^{p^{i'}}+ (\gamma u)^{p^{j'}}
= u A(\beta,\gamma),
\]
where $A$ is some function of $\beta,\gamma$. Then $k$ must divide 
$i'$ and $j'$; set $i'=uk$ and $j'=vk$ ($0\leq u<v\leq r-1$). 
Since $F$ maps to a subfield $\FF_{p^k}$,
the following must  be satisfied for all $x$:
\begin{eqnarray*}
	(\beta x^{p^{uk}} + x^{p^{vk}})^{p^k} - \beta x^{p^{uk}} - x^{p^{vk}}=0\\
	\beta^{p^{k}} x^{p^{(u+1)k}} + x^{p^{(v+1)k}}- \beta x^{p^{uk}} - x^{p^{vk}}=0.
	\end{eqnarray*}
	Hence, the exponents $\{p^{(u+1)k},p^{(v+1)k},p^{uk},p^{vk}\}$ cannot be two by two
distinct. This forces $u=v+1\pmod{r}$ and further  $v=u+1\pmod{r}$. This implies $u=u+2\pmod{r}$
showing that the only solution is  $u=0$ with $r=2$ and $v=r-1=1$ (using also $0\leq u<v\leq r-1$).
Finally, we must have
\[
\beta^{p^{k}} x^{p^k} + x- \beta x - x^{p^{k}}=
x^{p^k}\left(\beta^{p^{k}}-1\right)-x(\beta-1)=0,~~\mbox{for all $x$},
\]
which implies $\beta=1$ so that $F(x)=T^{2k}_k(x)$ completing the proof.
\end{proof}

Any function $f:\FF_{p^n}\rightarrow \FF_{p^k}$, $n=rk$, can be expressed
as $f(x)=T^n_k(P(x))$, where $P$ is some polynomial in $\FF_{p^n}[x]$.
Note that this representation is not unique. In the rest of this section
 we analyze
the case when $P$ has a single term, the cases with several terms being significantly more complicated.
The following result further refines the choice of $d$ for 
$f(x)=T^n_k(\beta x^d)$.
 We denote by $wt_H(d)$ the Hamming weight of $d$ which is the number of nonzero components in the $p$-adic expansion of integer $d$.
\begin{proposition}\label{prop:tracemon1}
	The function $f(x)=T^n_k(\beta x^d)$, $\beta \in \FF^*_{p^n}$,
 can have a linear translator only if  $wt_H(d) \in \{ 1, 2 \}$. 
When  $wt_H(d)=2$, then $d$ must be equal to $p^j(1+p^i)$ for some
$0\leq i,j\leq n-1$, $i\not\in\{0,n/2\}$. In particular, $f(x)=T^n_k(\beta x^{2p^j})$ cannot have linear translators.
\end{proposition}
\begin{proof}
In \cite[Theorem 5]{ChaKyu-fq},  it was proved that the function 
$T^n_1(\beta x^d)$ can have a linear structure
 only if $wt_H(d)\in \{1,2\}$.  Especially, when $wt_H(d)=2$ then $d=p^j(1+p^i)$ for some
$0\leq i,j\leq n-1$, $i\not\in\{0,n/2\}$.

Suppose now  that the function $f(x)=T^n_k(\beta x^d)$ has a $b$-translator $\gamma$. Then,
	\begin{eqnarray*}
	T^n_1\left(\beta (x+u\gamma)^d-\beta x^d \right) &=&  
T^k_1\left(T^n_k (\beta(x+u\gamma)^d-\beta x^d) \right)\\
	&=& T^k_1(bu).
	\end{eqnarray*}
	If we now fix $u\in \FF_{p^k}$, then $u\gamma$ becomes the 
$T^k_1(bu)$-linear structure of $T^n_1(\beta x^d)$, which gives the result. 
 In particular,  the function $T_1^n(\beta x^{2p^j})$ (corresponding to  $i=0$
 in $d=p^j(1+p^i)$) cannot have linear translators.	
\end{proof}

The following result was mentioned by Kyureghyan in \cite{kyu}.

\begin{lemma}\label{le:lin1}
	Let $f$ be an affine  function from $\FF_{p^n}$ to  $\FF_{p^k}$ given by 
	 $f(x)=T^n_k(\beta x)+a$, where $\beta\in\FF_{p^n}$ and $a\in\FF_{p^k}$.
	Then, any $\gamma\in\FF_{p^n}$ is a $b$-translator of $f$, with 
	$b=T^n_k(\beta\gamma)$.
\end{lemma}
\begin{proof}
	For any $\gamma\in\FF_{p^n}$ we have
	\[
	f(x+u\gamma)-f(x) = T^n_k(\beta (x+u\gamma))-T^n_k(\beta x)= 
	T^n_k(\beta u\gamma)=uT^n_k(\beta\gamma),
	\]
	for all $u\in\FF_{p^k}$ and $x\in\FF_{p^n}$.
\end{proof}
The next result regards the existence of linear translators for the trace
 of quadratic monomials which in general contains $r$ polynomial terms for $n=rk$.

\begin{lemma}\label{le:lt}
	Let $n=rk$ and $f(x)= T^n_k(\beta x^{p^i+p^j})$, where $i < j$.
	Then, $f$ has a derivative independent of $x$, that is,
 $f(x+u\gamma)-f(x)=T^n_k (\beta (u\gamma)^{p^i+p^j} )$
	for all $x\in \FF_{p^n}$,
	all $u\in \FF _{p^k}$, if and only if $\beta,\gamma \in \FF_{p^n}$ are related through, 
	\begin{equation}\label{eq:quadtrace}
	\beta \gamma^{p^{i+lk}}+\beta^{p^{(r-l)k}}\gamma^{p^{i+(r-l)k}}=0,
	\end{equation}
	where $0 < l < r$ satisfies $j=i+kl$. \\
	In particular, if $\beta \in \FF_{p^k}$ then  
$f(x+u\gamma)-f(x)=\beta T^n_k ((u\gamma)^{p^i+p^{i+kl}} )$ 
	 if and only if  $\gamma ^{p^{2kl}-1}=-1$, which requires $\frac{r}{\gcd(r,2l)}$ is even when $p>2$.
\end{lemma}
\begin{proof}
For $f(x)= T^n_k(\beta x^{p^i+p^j})$, we have
	\begin{eqnarray*}
	f(x+u\gamma)-f(x) &=& T^n_k\left(\beta(x+u\gamma)^{p^i+p^j}\right) -T^n_k \left(\beta x^{p^i+p^j}\right) \\
	&=&T^n_k \left(\beta x^{p^i}(u\gamma)^{p^j} + \beta x^{p^j}(u\gamma)^{p^i} +\beta(u\gamma)^{p^i+p^j} \right) \\
	&=&T^n_k \left(\beta x^{p^i}(u\gamma)^{p^j}\right) + T^n_k \left(\beta x^{p^j}(u\gamma)^{p^i}\right) +
	T^n_k \left(\beta (u\gamma)^{p^i+p^j} \right).
	\end{eqnarray*}
	The above expression will be independent of $x$ if and only if $T^n_k (\beta x^{p^i}(u\gamma)^{p^j}) =
	-T^n_k (\beta x^{p^j}(u\gamma)^{p^i})$, for all $x\in \FF_{p^n}$ and all $u\in \FF_{p^k}$. 

	 We analyze this equation in terms of the congruence $i \equiv j \pmod{k}$. If $i \not \equiv j \pmod{k}$, it follows that all the exponents are pairwise different. Therefore, all the coefficients must equal 0
and so either  $\beta = 0$ or  $\gamma= 0$. But  $\gamma$ cannot be 0, following from Definition~\ref{de:tr}, and  $\beta$ cannot be 0, since then $f(x) = 0$.

It follows that $i \equiv j \pmod{k}$,  thus $j=i+kl$
	for some $0< l<r$.  Note that we exclude the case $l=0$.
Indeed, in this case, $f$ is linear for  $p=2$ and $f(x)=x^{2p^i}$ for $p>2$,
a function which cannot have a linear translator by Proposition \ref{prop:tracemon1}.
Therefore, we have
	\begin{eqnarray}\label{eq:bin}
	f(x+u\gamma)-f(x) &=& T^n_k \left(\beta x^{p^i}(u\gamma)^{p^{i+lk}}\right) +
	T^n_k \left( \beta x^{p^{i+lk}}(u\gamma)^{p^i}\right) +
	T^n_k \left(\beta (u\gamma)^{p^i+p^{i+lk}} \right)\nonumber\\
	&=&  T^n_k \left(\beta x^{p^i}(u\gamma)^{p^{i+lk}}+
\beta^{p^{(r-l)k}} x^{p^{i}}(u\gamma)^{p^{i+(r-l)k}}+
 \beta(u\gamma)^{p^i+p^{i+lk}} \right)\nonumber\\
	&=& u^{p^i} T^n_k \left(x^{p^{i}} \left(\beta \gamma^{p^{i+lk}}+
\beta^{p^{(r-l)k}}\gamma^{p^{i+(r-l)k}}\right)\right)\nonumber\\ 
&& + ~u^{2p^i}T^n_k \left(\beta \gamma^{p^i+p^{i+lk}}\right) .
	\end{eqnarray}
	Thus, we must have  
\[
\beta \gamma^{p^{i+lk}}+\beta^{p^{(r-l)k}}\gamma^{p^{i+(r-l)k}}=0,
\]
to eliminate $x$.

In particular, if $\beta \in \FF_{p^k}$ then the above condition reduces to $\gamma ^{p^{2lk}-1} =-1$, which
	for $p$ odd has a solution exactly when $\frac{n}{\gcd(n,2kl)}=\frac{r}{\gcd(r,2l)}$ is even (see \cite[Claim 4]{ChaKyu-fq}, for instance).
	
\end{proof}
\begin{remark} It can be easily verified that \[
T^n_k(\beta x^{p^i+p^j})=\left(T^n_k(a x^{1+p^{j-i}} )\right)^{p^i},~a=\beta^{p^{n-i}}, ~j>i.
\]
Thus, alternatively,  one can consider the mapping $x\mapsto T^n_k(a x^{1+p^{s}})$. 

\end{remark} 
The result below  specifies further the existence of translators for quadratic trace monomials. 

\begin{theorem}\label{th:lt}
	Let $n=rk$ and  $f(x)=T^n_k(\beta x^{p^i+p^{j}})$, where $r>1$ and $j=i + kl$ for some $0 <l <r$. Assume that $\gamma\in \FF^*_{p^n}$ is a $b$-translator of $f$, where $b=T^n_k (\beta \gamma^{p^i+p^{i+lk}})$. Then :
	\begin{enumerate}[i)]
	\item If $p=2$ the condition (\ref{eq:quadtrace}) in Lemma \ref{le:lt} must be satisfied and either 
\[
b=T^n_k (\beta \gamma^{2^i+2^{i+lk}})~~\mbox{and $i=sk-1$ for some $0<s\leq r$,}
\]
 or $b=0$. In particular, if $\beta \in \FF_{2^k}$ then $\gamma=1$ is a
 $0$-translator of $f$ if $r$ is even and $\gamma=1$ is a $\beta$-translator
 if $r$ is odd, where in the latter case $i=sk-1$.
	\item If $p >2$  we necessarily have $b=0$. In  particular,
 if $\beta \in \FF_{p^k}$ then $n$ is even and $\gamma$ must satisfy
 $\gamma^{p^{2kl}-1}=-1$ and $Tr_k^n(\gamma^{1+p^{lk}})=0$.
	\end{enumerate}
\end{theorem}
\begin{proof}
	If (\ref{eq:quadtrace}) is satisfied then, from (\ref{eq:bin}),
	
	\begin{eqnarray*}
	f(x+u\gamma)-f(x) &=& u^{2p^i}T^n_k \left(\beta \gamma^{p^i+p^{i+lk}}\right).
	\end{eqnarray*}
	For $f$ to have linear translators, we  either have $u^{2p^i}=u$ or 
	$T^n_k(\beta \gamma^{p^i+p^{i+lk}})=0$. 
	
 $i)$ Let $p=2$. The condition $u^{2p^i}=u$ gives $2^{i+1}\equiv 1 \pmod{2^k-1}$,
 which implies $i=sk-1$, for some $0<s\leq r$. This follows from the fact that $2^k-1 \mid 2^{i+1}-1$ if and only if $k \mid i+1$. Otherwise, if $T^n_k(\beta \gamma^{2^i+2^{i+lk}})=0$ then $\gamma$ is a $0$-translator. 
 
  In particular, if $\beta \in \FF_{2^k}$ then $\gamma=1$ is a solution
 to (\ref{eq:quadtrace}).  Then,  
$$b=\beta T^n_k(\gamma^{2^i+2^{i+lk}})=\beta T^n_k(1)=0$$ if $r$ is even and
 $b=\beta$ for odd $r$ where additionally $i=sk-1$ as above.	

	$ii)$ For $p >2$ we have $2p^{i}\equiv 1 \pmod{p^k-1}$, which implies 
  $2p^i=1+s(p^k-1)$, for some $s$.
	Since $p$ is odd the left-hand side of the equation
 is even and the right-hand side is odd,
	which is impossible. The only remaining option  for $\gamma$ is
 to be a $0$-translator. 
 
 In particular, if $\beta \in \FF^*_{p^k}$, then by Lemma~\ref{le:lt}, $\frac{n}{\gcd(n,2kl)}=\frac{r}{\gcd(r,2l)}$ is even and thus $n$ must be even. Furthermore, (\ref{eq:quadtrace}) reduces to $\gamma^{p^{2kl}-1}=-1$ and the fact that $b=0$ implies 
\[
T^n_k(\beta \gamma^{p^i+p^{i+lk}})=\beta 
\left(T^n_k(\gamma^{1+p^{lk}})\right)^{p^i}=0.
\]
	
\end{proof}
\begin{remark}\label{open:pr1} The existence of translators for
 $f(x)=Tr_k^n(\beta x^{p^i+p^j})$ is more easily handled when
 $\beta \in \FF_{p^k}$. For $\beta \in \FF_{p^n}$    general solutions
 to (\ref{eq:quadtrace}) satisfying at the same time the other conditions
 seem to be difficult to specify explicitly. 
Theorem~\ref{th:lt} may also induce some non-existence results as well,
 which however requires further analysis. 
\end{remark}

The next corollary follows directly from Theorem \ref{th:main} and 
\ref{th:lt}.  

\begin{corollary}\label{cr:quad}
Let $p=2$, $n=rk$, $f(x)=T^n_k(\beta x^{p^{sk-1}+p^{(s+l)k-1}})$ for
 some $0<l<r$, $0<s\leq r$, and let
	$\gamma$ satisfy (\ref{eq:quadtrace}) in Lemma \ref{le:lt}.
 Then $$L(x)+L(\gamma )h\left( T^n_k(\beta x^{p^{sk-1}+p^{(s+l)k-1}})\right),$$
 where $L$ is a
	$\FF_{p^k}$-linear permutation on $\FF_{p^n}$ and
 $h:\FF_{p^k}\rightarrow\FF_{p^k}$, is a permutation if and
	only if $g:u\mapsto u+T^n_k ( \beta \gamma^{p^{sk-1}+p^{(s+l)k-1}})h(u)$
 permutes $\FF_{p^k}$.
\end{corollary}

\section{Compositional  inverses}\label{se:inv}
The main goal of this paper is to show that a lot of permutations,
and some related structures
can be derived from Theorem \ref{th:main}. 
In this section, we focus on the compositional inverses of these permutations. 
A similar initiative was taken in \cite{TuxaWang} where other classes of permutations (not of the form (\ref{eq:specclass})) were analyzed with respect to their inverses. Related  to compositional inverses of permutations of the form (\ref{eq:specclass}), we mention  Corollary 3.8 in \cite{TuxaWang} which states that given $\gcd (n,k)=d>1,
s(q^k-1) \equiv 0 \mod (q^n-1), \delta \in \mathbb{F}_{q^n}$, the function $f(x) = x + (x^{q^k}-x+\delta)^s$ permutes
$\mathbb{F}_{q^n}$ and its inverse is $f^{-1}(x)=x-(x^{q^k}-x+\delta)^s$.

\begin{definition}
Let $F$ be any function  over  $\FF_{p^n}$. For any $t\geq 1$, the function
$$F_t(x)=\underbrace{F\circ \dots \circ F}_{t}(x)$$ is said to be
the {\it t-fold composition} of $F$ with itself.
\end{definition}
In \cite[Section 4]{kyu}, the author studied the functions
$F:x\mapsto x+\gamma f(x)$, {\it i.e.,} with notation of Definition \ref{de:tr},
the function $h$ being the identity. Several results in \cite{kyu}, regarding the compositional inverses,
hold for such $F$ (only). Henceforth,  we attempt 
to specify compositional inverses when $h$ is not the identity.  

\begin{lemma}\label{le:inv1}
Let $n=rk$, $k>1$. 
Let $f:\FF_{p^n}\rightarrow \FF_{p^k}$, 
$h:\FF_{p^k}\rightarrow \FF_{p^k}$
and $b\in\FF_{p^k}$. Define
$$F(x)=x+\gamma h(f(x)),~\gamma\in\FF_{p^n}^*$$
where $\gamma$ is a $b$-linear translator of $f$.
Then 
\[
F_2 (x)=x+\gamma h(f(x))+\gamma h\left(bh(f(x))+f(x)\right).
\]
\end{lemma}
\begin{proof}
\begin{eqnarray*}
F\circ F(x) &=& F\left(x+\gamma h(f(x))\right)\\
&=&  x+\gamma h(f(x))+\gamma h\left(f(x+\gamma h(f(x)))\right)\\
&=&  x+\gamma h(f(x))+\gamma h\left(bh(f(x))+f(x)\right),
\end{eqnarray*}
since $f(x+\gamma h(f(x)))=b h(f(x))+f(x)$ for all $x$.
\end{proof}
\begin{proposition}\label{pr:inv1}
Notation is as in Lemma \ref{le:inv1}.
If $b=0$ then $F_p(x)=x$ so that 
\[
F^{-1}(x)=F_{p-1}(x)=x+(p-1)\gamma h(f(x)).
\] 
In particular, $F$ is an involution when $p=2$.
\end{proposition}
\begin{proof}
Assume that $b=0$. In this case, $F$ is a permutation for any 
$h$ (from Theorem  \ref{th:main}), 
so that its  compositional  inverse $F^{-1}$  exists. We get 
from  Lemma \ref{le:inv1}:
\[
F\circ F(x)=x+2~\gamma h(f(x)).
\]
Assume that $F_{j-1}(x)= x+(j-1)\gamma h(f(x))$.  We have   for $2<j\leq p$:
\begin{eqnarray*}
F_j(x)  &=& F\circ F_{j-1}(x)=F_{j-1}(x)+\gamma h(f(F_{j-1}(x)))\\
  &=& x+(j-1)\gamma h(f(x))+\gamma h(f(x))=x+j\gamma h(f(x)),
\end{eqnarray*}
since $f(x+(j-1)\gamma h(f(x)))=f(x)$, for all $x$.
Thus we get $F_p(x) =x$, for all $x$, for $j=p$. Moreover if $p=2$
then $F^{-1}=F$.
\end{proof}

Thus, according to Proposition \ref{pr:inv1}  a large set of permutations
can be obtained whose compositional inverse is known as illustrated
below.
\begin{corollary}\label{cr:inv}
 Let $f:\FF_{p^n} \rightarrow \FF_{p^k}$, $n=rk$, $f(x)=T^n_k(\beta x)$. Choose
$\beta ,\gamma\in\FF_{p^n}^*$ such that $T^n_k(\beta\gamma)=0$.
Let $L$ be any $\FF_{p^k}$-linear permutation.
Then the functions
\[
F(x)=L(x)+L(\gamma) h\left(T^n_k(\beta x)\right)
\]
are  permutations for any $h:\FF_{p^k}\rightarrow \FF_{p^k}$.
Moreover 
\[
F^{-1}(x)=L^{-1}(x)+(p-1)L(\gamma) h\left(T^n_k(\beta(L^{-1}(x)))\right).
\]
If $p=2$ and $L(x)=x$, then $F$ is an involution, {\it i.e.,} $F^{-1}=F$.
\end{corollary}
\begin{proof}
From Lemma \ref{le:lin1}, $\gamma$ is a $0$-translator of $f$ if and only if
 $T^n_k(\beta\gamma)=0$. So, from Theorem \ref{th:main}, 
 $F$ is a permutation for any $\FF_{p^k}$-linear 
permutation $L$ and  for any $h$. Further, set $G(x)=x+\gamma h(f(x))$
so that $F=L\circ G$. Then $F^{-1}=G^{-1}\circ L^{-1}$, where,
from Proposition \ref{pr:inv1},
$$G^{-1}(x)=G_{p-1}(x)=x+(p-1)\gamma h\left(T^n_k(\beta x)\right).$$  
Moreover if  $p=2$ and $L(x)=x$,  then 
$F^{-1}(x)=G^{-1}(x)$ with $G^{-1}(x)=G(x)$.
\end{proof}

Taking $h$ linear we get a large set of  linear permutations. We illustrate this
in the binary case when $r=2$.
\begin{corollary}\label{cr:inv1}
Notation is as in Corollary \ref{cr:inv} with $n=2k$ and $p=2$. Assume that 
$L$ is a $\FF_{p^k}$-linear involution, {\it i.e.,}
$L(x)=ax+bx^{2^k}$ as defined by Lemma \ref{le:lin}. Then, for all 
$\beta ,\gamma\in\FF_{p^n}^*$ such that $T^n_k(\beta\gamma)=0$
and for any linear function $h$ the functions
\[
F(x)=L(x)+L(\gamma) h\left(T^n_k(\beta x)\right),
\]
are linear permutations of  $\FF_{p^n}$ and
\[
F^{-1}(x)=L(x)+L(\gamma) h\left(T^n_k(\beta(L(x)))\right).
\]
\end{corollary}
Note that for $p=3$ the compositional inverse is obtained by adding 
to $F$ its second term, as  shown in the example below.
\begin{example}
Let $p=3$, $n=3k$ and $a\in\FF_{3^k}$. 
\[
F(x)=x+\gamma(x^{3^{2k}}+x^{3^k}+x+a)^s,~T^{3k}_k(\gamma)=0.
\] 
Then, by applying Corollary \ref{cr:inv}, 
 $F$ is a permutation of $\FF_{3^n}$ for any integer $s$ in the
range $[1, 3^n-2]$.
Moreover 
\[
F^{-1}=x+2\gamma\left(x^{3^{2k}}+x^{3^k}+x+a\right)^s=F(x)+
\gamma\left(T^n_k(x)+a\right)^s.
\]
\end{example}

In Section \ref{se:trans}, it was proved that a function  
$f:\FF_{p^n}\rightarrow \FF_{p^k}$, $p$ odd, defined by 
$f(x)=T^n_k(x^{p^i+p^{i+\ell k}})$, 
can have a $b$-translator
for $b=0$ only (see Theorem \ref{th:lt}). Based on this, we are able to derive 
a class of permutations of degree at least $2$ whose compositional inverse
is known.
\begin{corollary}\label{cor:4}
Let $p$ be an odd prime, $n=rk$ and $\ell$ be a positive integer 
such that $r/\gcd(r,2\ell)$ is even. Let $f(x)=T^n_k(x^{p^i+p^{i+\ell k}})$, 
where  $0\leq i\leq k-1$. Let $\gamma\in\FF^*_{p^n}$ such that
\[
\gamma^{p^{2k\ell}-1}=-1~~\mbox{and}~~T^n_k(\gamma^{1+p^{\ell k}})=0.
\] 
Then 
\[
x\mapsto L(x)+L(\gamma )h\left(T^n_k(\beta x^{p^i+p^{i+\ell k}})\right)
\] 
is a permutation of $\FF_{p^n}$, for 
any  $\FF_{p^k}$-linear permutation $L$  and  any
$h:\FF_{p^k}\rightarrow\FF_{p^k}$.
Moreover if $F(x)=x+\gamma h\left(T^n_k(x^{p^i+p^{i+\ell k}})\right)$
then 
\[F^{-1}(x)=x+\gamma (p-1) h\left(T^n_k(x^{p^i+p^{i+\ell k}})\right).\]
\end{corollary}
\begin{proof}
From Theorem \ref{th:lt}, $\gamma$ is a  $0$-linear translator  of $f$ 
 if and only if  $T^n_k(\gamma^{1+p^{\ell k}})=0$.
Further, we apply Theorem \ref{th:main} and Proposition \ref{pr:inv1}.
\end{proof}
We previously considered functions with a zero translator, {\it i.e.,  $b=0$},
to obtain permutations with their compositional inverses.
When $b\ne 0$, other permutations with their compositional inverses 
can be obtained.
In this case however, it seems that the definition of the function
 $h$ has to be specified.
The idea is to determine $h$ such that 
\[
h(f(x))+h(bh(f(x))+f(x))=g(x), ~b\ne 0,
\]
(by using Lemma \ref{le:inv1})  
where $g$ allows us to compute easily the $t$-fold composition of
 $F$ with itself.
We illustrate our  purpose by constructing involutions for any odd $p$.
\begin{proposition}\label{pr:inv2}
Notation is as in Lemma \ref{le:inv1}.   Let $p$ be an odd prime. 
Assume that $\gamma$ is a $b$-linear translator of $f$
where $b\ne 0$.  Set $h(x)=\lambda x$ where $\lambda \in\FF^*_{p^k}$ and 
$\lambda\ne -b^{-1}$. Then  the function $F$,
\[
F(x)=x+\gamma \lambda f(x),
\]
permutes $\FF^*_{p^n}$. Moreover, if $\lambda=-2b^{-1}$  then $F$ is 
an involution.
\end{proposition} 
\begin{proof}
From Theorem \ref{th:main}, $F$ is a permutation, since 
\[
\ell(u)=u+bh(u)=u(1+\lambda b)~~\mbox{for  $u\in\FF_{p^k}$}; 
\]
so $\ell$ is a permutation because $\lambda \ne -b^{-1}$ by hypothesis.
Moreover
\[
h(f(x))+h(bh(f(x))+f(x))=2\lambda f(x)+b\lambda^2 f(x)=
\lambda f(x)(2+b\lambda).
\]
From Lemma \ref{le:inv1}, we get $F\circ F(x)=x$  if and only if $2+b\lambda=0$.
\end{proof}
\section{Relation with complete permutations}\label{se:cplte}
The concept of complete permutations is of crucial importance for non-zero linear translators $b$ in terms of Theorem~\ref{th:main}, since the main condition there was that $u\mapsto u+bh(u)$  
permutes $\FF_{p^k}$.  
\begin{definition}
Let $h$ be a function over $\FF_{p^k}$. We say that $h$ is complete
with respect to $b$, or {\it b-complete}, when both $h$ and $u\mapsto u+bh(u)$  
permute $\FF_{p^k}$.
\end{definition}
Thus we can apply Theorem \ref{th:main} as follows:
\begin{theorem}\label{th:cplte}
Let $n=rk$, $k>1$. 
Let $f:\FF_{p^n}\rightarrow \FF_{p^k}$, 
$h:\FF_{p^k}\rightarrow \FF_{p^k}$, $\gamma\in\FF_{p^n}^*$
and $b\in\FF^*_{p^k}$ such that $\gamma$ is a $b$-linear translator of $f$.
 Let $L$ be a $\FF_{p^k}$-linear permutation on $\FF_{p^n}$.

If  $h$ is $b$-complete
then $F(x)=L(x)+L(\gamma)h(f(x))$ permutes $\FF_{p^n}$.
\end{theorem}
\begin{proof}
To say that  $h$ is $b$-complete is to say that  both $h$ and 
$u\mapsto u+bh(u)$  permute $\FF_{p^k}$. We apply Theorem \ref{th:main}
assuming that $h$ is a permutation.
\end{proof}
The characterizarion of complete permutations, especially monomials,
is currently discussed in many works (see for instance 
\cite{BasZin,TuZeHu14,WuLiHeZh14} and  references). 
New permutations could be obtained
while $h$ is not bijective, as in the next example.

\begin{example}
 Let $p=3$ and $h$ be the function on $\FF_{p^3}$ defined by $h(x)=x^{p^2+p+2}$.
By \cite[Theorem 6]{BasZin} we know those $b\in\FF_{p^3}$ such that 
$u\mapsto u+bh(u)$  permutes $\FF_{p^3}$. Thus, we can apply
 Theorem \ref{th:main} for any $n=3r$ and for any such $b$.
Let $\gamma\in\FF_{p^n}$ and
\[
f~:~x\in\FF_{p^n}~\mapsto~x+x^{p^3}+\dots+ x^{p^{(r-1)3}}\in\FF_{p^3}.
\]
Then, for any $u\in\FF_{p^3}$
\[
f(x+u\gamma)-f(x)=T^{3r}_3(u\gamma)=uT^{3r}_3(\gamma).
\]
 Thus, we choose $\gamma$ such that $b=T^{3r}_3(\gamma)$ is suitable,
according to the results of \cite{BasZin}.
Then we obtain a new permutation $F$, for any $\FF_{p^3}$-linear permutation $L$.
In particular for $L(x)=x$:
\[
F(x)~:~x\mapsto x+\gamma \left(T^{3r}_3(x)\right)^{p^2+p+2}
\]
is a permutation of $\FF_{p^n}$.
Another example is $L(x)=ax+x^{p^3}$, where $a\in\FF_{p^n}$ and $-a$ are not in the
 image set of $x\mapsto x^{p^3-1}$. Then
\[
F(x)=ax+x^{p^3}+L(\gamma)\left(T^{3r}_3(x)\right)^{p^2+p+2}
\]
is a permutation over $\FF_{p^n}$.
\end{example}
A set of trinomials which are $1$-complete over $\FF_{2^{3m}}$ is proposed
in \cite[Theorem 4]{TuZeHu14}. We give here a slightly different version
of this result.
\begin{theorem}\label{th:cpp}
For any  $\nu\in\FF_{2^m}\setminus\{0,1\}$, the trinomial
\[
h(x)=x^{2^{2m}+1}+x^{2^m+1}+\nu x
\]
is complete over $\FF_{2^{3m}}$ with respect to any 
$b\in\FF_{2^m}\setminus\{0,\nu^{-1}\}$.
\end{theorem}
\begin{proof}
It is proved in \cite{TuZeHu14}  that $h$ is a permutation of 
$\FF_{2^{3m}}$ for any such $\nu$.
Thus $x\mapsto bh(x)$ is also a permutation. If
$b\in\FF_{2^m}\setminus\{0,\nu^{-1}\}$ then $b\nu+1\in\FF_{2^m}\setminus\{0,1\}$.
 So we have
\[
g(x)=b\left(x^{2^{2m}+1}+x^{2^m+1}+\nu x\right)+x=bh(x)+x,
\]
where $h$ and $g$ are both bijective.
\end{proof}
Applying Theorem \ref{th:main}, we obtain directly the following
class of permutation.
\begin{corollary}
Let $n=rk$ with $k=3m$. Denote by $L$ any $\Fk$-linear permutation on 
$\F$. Let  $f:\F\rightarrow \Fk$ such that $f$ has a $b$-translator
 $\gamma\in\F^*$  with $b\in\FF^*_{2^m}$.
Then the functions 
\[
x\mapsto L(x)+L(\gamma)\left((f(x))^{2^{2m}+1}+(f(x))^{2^{m}+1}+\nu(f(x))\right)
\]
permute $\F$ for all $\nu\in\FF_{2^m}\setminus\{0,1,b^{-1}\}$.
\end{corollary}
\section{A special class of permutations}\label{se:spec}
There is currently a lot of work related to  the functions over $\FF_{p^n}$ of type
\EQ\label{typ}
F~:~x\mapsto (f(x)+\delta)^s+L(x),~\delta\in\FF^*_{p^n},
\EN
where $f$ is linear, $s$ is any integer and $L$ is a linearized polynomial
in $\FF_{p^n}[x]$ (see \cite{TuZeLiHe15},\cite{TuZeJi15} and \cite{YuDiWaPi} 
for the most recent articles, and their references).
The problem is {\it to determine some $(\delta,s,L)$ such that
$F$ is a permutation.} To apply directly Theorem \ref{th:main}, we 
take $L(x)=x$ and specific functions $f$. According to our previous
results and thanks to Theorem \ref{th:main}
we can treat  some cases directly. Note that $\delta$ must be  in the
 image set of $f$ to apply Theorem \ref{th:main}.

\begin{proposition}\label{pr:perm}
Let $n=2k$, $F(x)=\gamma(f(x))^s+x$ where $f(x)=x^{p^k}+x+\delta$
with $\delta\in\FF_{p^k}$. Set $b=T^n_k(\gamma)$. Then 
\begin{itemize}
\item If  $b=0$ then  $F$ is a permutation over $\F$ for any
$s$ as well as 
\[
x\mapsto L(\gamma)(f(x))^s+ L(x) ~\mbox{where $L$ 
is an $\FF_{p^k}$-linear permutation.}
\]
\item When $b=0$, $F^{-1}(x)=x+(p-1)\gamma(f(x))^s$. Notably,
 $F$ is an involution if and only if $p=2$.
\item  When  $b\ne 0$, 
one  can apply Theorem \ref{th:main}  if and
only if $u\mapsto u+bu^s$ permutes $\FF_{p^k}$. It is especially the case when $u\mapsto u^s$ is $b$-complete.
\end{itemize}
\end{proposition}
\begin{proof}
First, we have from Lemma \ref{le:lin1}:
\[
f(x+\gamma u)-f(x)=u(\gamma^{p^k}+\gamma),
\]
 for all $u\in\FF_{p^k}$ and all $x$. Thus $\gamma$  is a 
$b$-translator of $f$, with $b=\gamma^{p^k}+\gamma$. 

To have $b=0$ is always possible.  When $p=2$ we take 
$\gamma\in \FF_{p^k}$. When $p$ is odd it is known that $\gamma^{p^k-1}=-1$ has
a solution in $\FF_{p^n}$ as soon as $n/\gcd(n,k)$ is even
(see \cite[Claim 4]{ChaKyu-fq}, for instance). Here we have $2k/\gcd(2k,k)=2$.
For such $\gamma$, we can apply 
Theorem \ref{th:main} for any $s$. Moreover, the inverse of $F$ is
obtained by applying Proposition \ref{pr:inv1}. According to Theorem \ref{th:cplte}, we can apply 
Theorem \ref{th:main} in particular when $u\mapsto u^s$ is
$b$-complete.
\end{proof}

Our purpose is to contribute to the current works on
polynomials of type (\ref{typ}).
Ge\-nerally, to prove that $F$ is a permutation is easier when 
$\delta$ is in a subfield and $f$ has its image in this subfield.
In the next subsections we study specific polynomials, taking
 $\delta\in\FF_{p^n}$ where $n=2k$.   
The results presented by Propositions \ref{pr:var1} and \ref{pr:var2}
(and then Corollary \ref{cr:spe})
are partly already known. The necessary and sufficient condition of bijectivity can be obtained by using the AGW 
criterion. More precisely,  we give here instances and applications of the
following result which is  a direct consequence of 
\cite[Theorem 5.1]{AkGhWa11}.
 We first give  the version of
\cite[Proposition 5.9]{AkGhWa11} that we need in our context.
\begin{proposition}\label{pr:ak}
Let $L$ be an $\FF_{p^k}$-linear polynomial which permutes $\FF_{p^k}$
and $g,h:\FF_{p^n}\rightarrow \FF_{p^n}$,
where $h(x^{p^k}-x)\in\FF_{p^k}^*$.

Then the function $x\mapsto h(x^{p^k}-x)L(x)+g(x^{p^k}-x)$ is a 
permutation of $\FF_{p^n}$ if and only if
\[
x\mapsto h(x)L(x)+g(x)^{p^k}-g(x)~\mbox{ permutes } \mathcal{J}=\{y^{p^k}-y|y\in\FF_{p^n}\}.
\]
\end{proposition}
We propose another way of proving the bijectivity in  Propositions
 \ref{pr:var1} and \ref{pr:var2}.
Our main purpose is  to use  the component functions of
  $F$ explicitly relying on  the following criterion:
{\it $F:\FF_{p^n}\mapsto\FF_{p^n}$ is a permutation if and only if all
 its component functions $F_\lambda(x)=Tr(\lambda F(x))$, $\lambda \in\FF_{p^n}^*$,
 are balanced \cite[Theorem 7.7]{LidNied}. } This approach may have independent significance for establishing permutation property of other classes of functions and may be useful in the analysis of the Walsh spectra of the component functions.
\subsection{Permutation polynomials for ${\bf p =2}$}
When $p=2$, to say that the component functions 
$F_\lambda(x)=Tr(\lambda F(x))$ of $F$ are {\it balanced} is to prove that
\begin{equation}\label{eq:cF}
A_\lambda =\sum_{x\in\F}(-1)^{Tr(\lambda F(x))}=0,~\forall~\lambda \in\F^*.
\end{equation}
\begin{proposition}\label{pr:var1}
Let $n=2k$ and $F:\F\rightarrow \F$ with $F(x)=x+(x+x^{2^k}+\delta)^s$,
where $\delta\in\F$ and $s$ is any integer in the range $[0,2^n-2]$. 
Notation $F_\lambda$ and $A_\lambda$ is defined above.
Let us define
\[
g~:~y\mapsto y+(y+\delta)^s+(y+\delta)^{2^ks}~~\mbox{ from $\Fk$ to $\Fk$}.
\]
Then we have:
\begin{description}
\item[(i)] $F$ is a permutation over $\F$ if and only if
the function $g$ is bijective. In parti\-cular, if $s$ satisfies 
$2^k s  \equiv s \pmod{2^n-1}$  then $F$ is a permutation.
\item[(ii)] The Boolean functions $F_\lambda$ are balanced for all
 $\lambda\not\in\Fk$.
 If $\lambda\in\Fk$  then 
\[
A_\lambda=2^k\sum_{y\in\Fk}(-1)^{T^k_1\left(\lambda g(y)\right)}.
\]
\end{description}
\end{proposition}
\begin{proof}
Note that $s=0,1$ are trivial cases. So we suppose that $s \geq 2$.
The item (i) comes directly from Proposition \ref{pr:ak}, by
taking (with its notation) $L(x)=x$, $g(x)=(x+\delta)^s$ 
and $h$ is the constant 
function equal to $1$. Note that in this case $\J=\Fk$.
Clearly, if $2^k s  \equiv s \pmod{2^n-1}$ then $g(y)=y$, and 
 thus $F$ is a permutation.

(ii) Now, it is easy to see that $F$ is affine on any
coset of $\Fk$: for $x=a+y$, $y\in \Fk$
\[
F(a+y)=y+a+(a+a^{2^k}+\delta)^s.
\]
 Let $\W$ be a set of representatives of these cosets.
Thus $\F=\cup_{a\in \W} (a+\Fk)$.
We have for any $\lambda\in\F^*$:
\begin{eqnarray*}
A_\lambda &=& \sum_{a\in\W}\sum_{y\in\Fk}(-1)^{Tr(\lambda F(y+a))}\\
&=&  \sum_{a\in\W}\sum_{y\in\Fk}(-1)^{Tr(\lambda (y+a+(a+a^{2^k}+\delta)^s))}\\
&=& \sum_{a\in\W}\sum_{y\in\Fk}(-1)^{T^k_1\left((\lambda+\lambda^{2^k})y +
T^{2k}_k(\lambda F(a))\right)}.
\end{eqnarray*}
We deduce that $A_\lambda=0$ for any $\lambda\not\in \Fk$, which means that 
$F_\lambda$ is balanced for all these $\lambda$. Now assume that
 $\lambda \in\Fk^*$.  Then
\[
A_\lambda=2^k\sum_{a\in\W}(-1)^{T^k_1\left(T^{2k}_k(\lambda F(a))\right)},
\]
where
\[
T^{2k}_k(\lambda F(a))=\lambda\left(a+a^{2^k}+(a+a^{2^k}+\delta)^s+
(a+a^{2^k}+\delta)^{2^ks}\right).
\]
Since $a\mapsto a+a^{2^k}$ is a bijection from $\W$ to $\Fk$,
to compute the values  $T^{2k}_k(\lambda F(a))$  is exactly to compute $\lambda g(y)$ for $y\in\Fk$.
Clearly, $A_\lambda=0$ for all $\lambda\in\Fk^*$ if and only if 
$g$ is bijective. 
\end{proof}
\begin{remark} In a recent article  \cite{TuZeJi15}, two classes of 
permutations $F(x)=x + (x+x^{2^k}+\delta)^s$ were proposed  for $s$ of the form 
$s = i (2^k \pm 1) +1$. More precisely, it was shown that $F$ is a permutation
 for $s = 2(2^k - 1) +1=2^{k+1}-1$ and for $s \in \{2^k +2, 2^{2k-1} +2^{k-1} +1, 2^{2k} - 2^k - 1\}$ when $s=i (2^k + 1) +1$. 
The above result covers the case $s=i(2^k+1)$ for any $i \in [0,2^k-2]$,
since in this case $s(2^k-1)\equiv 0\pmod{2^n-1}$.
\end{remark}
It is also of interest to establish whether for $s=2^i$, for 
$i=0, \ldots, n-1$, the linearized polynomial $F(x)$ is a permutation.
 An immediate consequence of Proposition~\ref{pr:var1} is the following. 
\begin{corollary} \label{cor:linperm} Using the same notation as in 
Proposition \ref{pr:var1},
 if $s=2^i$ then $F(x)=x+(x+x^{2^k}+\delta)^s$ is a linearized permutation for
 any $\delta \in \FF_{2^n}$ and any $i=0,\ldots,n-1$. 
\end{corollary}
\begin{proof}
Since $F$ is a permutation if and only if $g(y)=y+T^{2k}_k((y+\delta)^s)$ is
 a permutation over $\Fk$, then for $s=2^i$ we have 
\[
g(y)=y + T_k^n(y^{2^i}) + T_k^n(\delta^{2^i})=y + T_k^n(\delta^{2^i})
\]
 which is clearly a permutation. 
\end{proof}
Another direct consequence of Proposition~\ref{pr:var1} is the following result. 
\begin{corollary} \label{cor:delta} Using the same notation as in 
Proposition~\ref{pr:var1}, if $\delta \in \Fk$ then $F(x)=x+(x+x^{2^k}+\delta)^s$
 is a  permutation for any $s \in [0,2^k-2]$. 
\end{corollary}
\begin{proof}
If $\delta \in \Fk$ then $(y+\delta)^s \in \Fk$ since $y \in \Fk$ 
so that $g(y)=y + T_k^n((y+\delta)^s)=y$, which is a permutation and so is $F$
 regardless of the choice of $s$.
\end{proof}
\begin{remark}  Corollary~\ref{cor:delta} also follows from 
 Proposition~\ref{pr:perm} by noting that  in this case
 $b=0$, that is, $\gamma=1$ is a $0$-translator.
Recall that in this case $F$ is an involution for any  $\delta \in \Fk$.
 \end{remark}

\subsection{Permutation polynomials for odd $\bf p$}
Using the same technique, we deduce  slightly different results when $p$ is odd. 
For odd $p$, the function $F_\lambda$ is said
to be {\it balanced}  when
\EQ\label{eq:pp}
A_\lambda=\sum_{x\in\FF_{p^n}}\zeta_p^{Tr(\lambda F(x))}=0
\EN
where $\zeta_p$ is a $p$-th root of unity, {\it i.e.,} $\zeta_p=e^{2\pi i/p}$
for some $i$. Also, $F$ is a permutation over $\FF_{p^n}$ if and only if
(\ref{eq:pp}) holds for any $\lambda\in\FF^*_{p^n}$.
\begin{proposition}\label{pr:var2}
Let $p$ be an odd prime,
$n=2k$ and $F:\FF_{p^n}\rightarrow \FF_{p^n}$, 
$$F(x)=L(x)+(x^{p^k}-x+\delta)^s, ~\delta\in\FF_{p^n},$$
where $L\in \FF_{p^k}[x]$ is a linear permutation 
and $s$ is any integer in the range $[1,p^n-2]$. 
Let us define 
\[
G(y)=-L(y)+(y+\delta)^s-(y+\delta)^{p^ks}, ~y\in\FF_{p^n}.
\]
Then we have:
\begin{description}
\item[(i)] $F$ is a permutation over $\FF_{p^n}$ if and only if
the function $G$ permutes the subspace $\Sp=\{y\in\FF_{p^n}~|~T^n_k(y)=0\}$.
 In particular, if $s$ satisfies $p^k s  \equiv s \pmod{p^n-1}$ 
 then $F$ is a permutation.
\item[(ii)] The component functions $F_\lambda$ of $F$ are balanced for all
$\lambda\in \FF_{p^n}^*$  satisfying $T^n_k(\lambda)\ne 0$.
 If $T^n_k(\lambda)= 0$,  then 
\[
A_\lambda=p^k\sum_{y\in\Sp}\zeta_p^{T^k_1\left(\lambda G(y)\right)}.
\]
\end{description}
\end{proposition}
\begin{proof}
First, (i) comes directly from Proposition \ref{pr:ak}, by
taking (with its notation) $g(x)=(x+\delta)^s$ and $h$ is the constant 
function equal to $1$. Obviously $\Sp=\J$, since
$\J$ and $\Sp$  have the  same  cardinality $p^k$ and $\J\subset \Sp$ 
because $y=u^{p^k}-u$ satisfies $T^{2k}_k(y)=0$. Note that
$L(\Sp)=\Sp$ since
\[
L(y)+(L(y))^{p^k}=L(y+y^{p^k})=L(0)=0,~~\mbox{for any $y\in\Sp$}.
\]
If $s$ satisfies $p^k s  \equiv s \pmod{p^n-1}$, then $G(y)=-L(y)$
 implying that $F$ is a permutation since $L$ permutes $\Sp$ by assumption.

As in Proposition \ref{pr:var1}, $\W$ is a set of representatives 
of the $p^k$ cosets of $\FF_{p^k}$. Recall that $F_ \lambda(x) =Tr(\lambda F(x))$.
We have for any $\lambda\in\FF^*_{p^n}$:
\[
A_\lambda = \sum_{x\in\FF_{p^n}}\zeta_p^{Tr(\lambda F(x))}= \sum_{a\in\W}\sum_{y\in\FF_{p^k}}
\zeta_p^{Tr(\lambda F(y+a))},
\]
where
\begin{eqnarray*}
Tr(\lambda F(y+a)) &=& Tr\left(\lambda (L(y+a)+(a^{p^k}-a+\delta)^s)\right)\\
&=& T^k_1\left(L(y)(\lambda+\lambda^{p^k})+T^{2k}_k(\lambda F(a))\right).
\end{eqnarray*}
Since $L \in \FF_{p^k}[x]$ is a permutation over $\FF_{p^n}$ and thus 
over $\FF_{p^k}$ as well, we deduce that $A_\lambda=0$ for any $\lambda$ 
such that  $\lambda+\lambda^{p^k} \ne 0$, {\it i.e.},
$F_\lambda$ is balanced for such $\lambda$. Further,  for  
$\lambda + \lambda^{p^k}=0$, thus $\lambda \in \Sp$, we get
\[
A_\lambda=p^k\sum_{a\in\W}\zeta_p^{T^k_1\left(T^{2k}_k(\lambda F(a))\right)},
\]
where
\begin{eqnarray*}
T^{2k}_k(\lambda F(a))  &=& (\lambda L(a))^{p^k}+\lambda L(a)+
T^{2k}_k(\lambda(a^{p^k}-a+\delta)^s)\\
&=& \lambda\left(L(a)-L(a^{p^k})+(a^{p^k}-a+\delta)^s-
(a^{p^k}-a+\delta)^{sp^k}\right)\\
&=& \lambda\left(L(a-a^{p^k})+(a^{p^k}-a+\delta)^s-
(a^{p^k}-a+\delta)^{sp^k}\right).
\end{eqnarray*}
Recall that  $\pm(z^{p^k}-z)\in \Sp$,
 for any $z\in\FF_{p^n}$. Moreover $\lambda s\in\FF_{p^k}$ for any $s\in\Sp$,
since 
\[
(\lambda s)^{p^k}=\lambda^{p^k}s^{p^k}=(-\lambda)(-s)=\lambda s.
\]
Therefore, $T^{2k}_k(\lambda F(a)) =\lambda B $ with
\EQ\label{eq:B}
B=L(a)+(a^{p^k}-a+\delta)^s-
\left(L(a)+(a^{p^k}-a+\delta)^{s}\right)^{p^k},
\EN
which satisfies  $T^{2k}_k(B)=0$, {\it i.e.,}  $B\in \Sp$. Clearly,
the function $a\mapsto a^{p^k}-a$ is a bijection from  $\W$
to $\Sp$.  Finally,  the function
\[
G(y)=-L(y)+(y+\delta)^s-(y+\delta)^{sp^k}, 
\]
can be viewed as  a function from the subspace $\Sp$ to itself 
and $\lambda G(y)\in\FF_{p^k}$. Consequently
\[
A_\lambda=p^k\sum_{y\in\Sp}\zeta_p^{T^k_1\left(\lambda G(y)\right)}.
\]
Note that $A_\lambda=0$ for any $\lambda\in\Sp$ if and only if $G$
 is a permutation of $\Sp$.
\end{proof}
\begin{corollary} \label{cor:sdelta}
Notation is as in Proposition \ref{pr:var2}. Assume that $T^n_k(\delta)=0$.
Then
\begin{itemize}
\item If $s$ is even then $F$ is a permutation of $\FF_{p^n}$ for
 any permutation $L$.
\item If $s$ is odd then $F$ is a permutation of $\FF_{p^n}$ if and only if
\[
y\mapsto L(y)-2(y+\delta)^s~~\mbox{is a permutation of $\Sp$.}
\]
\item If $s$ is even and $L(x)=x$,  then  we have   $F^{-1}(x)=F_{p-1}(x)$. 
\end{itemize}
\end{corollary}
\begin{proof}
As we noticed in the previous proof,  $L$ induces a permutation of $\Sp$.
The case $s$ even  was proved in \cite[Theorem 3.4]{YuDi}. Another proof
is simply derived from Proposition \ref{pr:var2}  by  observing that
\begin{eqnarray*}
G(y) &=& -L(y)+(y+\delta)^s-(-y-\delta)^{s}\\
&=& -L(y)+(y+\delta)^s-(-1)^s(y+\delta)^s
=-L(y).
\end{eqnarray*}
When $s$ is odd, we get $G(y)=-L(y)+2(y+\delta)^s$.  Now consider
\[
F(x)=x+(f(x))^s, f(x)=x^{p^k}-x+\delta,~\mbox{with $s$ even}.
\]
Note that $f(x)\in\Sp$ when $T^n_k(\delta)=0$, since $f(x)^{p^k}=-f(x)$.
Moreover,
\EQ\label{eq:S}
(f(x))^{sp^k}-(f(x))^s=(-f(x))^s-(f(x))^s=0
\EN
holds for any even $s$. To compute the inverse of $F$ we proceed as
in Section \ref{se:inv}. We have here
\EQ\label{eq:S0}
F\circ F(x) =   F(x)+\left(f(x+(f(x))^s)\right)^s,
\EN
where  $T^{2k}_k(f(x))=0$. Setting  $a=(f(x))^s$, we get
\begin{eqnarray*}
f(x+a)-f(x) &=& (x+a)^{p^k}-(x+a)+\delta-x^{p^k}+x-\delta\\
&=& a^{p^k}-a= 0,~~\mbox{from (\ref{eq:S})}.
\end{eqnarray*}
Hence, according to (\ref{eq:S0}),
\[
F_2(x)=F(x)+\left(f(x)\right)^s = x+2(x^{p^k}-x+\delta)^s.
\]
Further, for $j>2$, assuming that $F_{j-1}(x)=x+(j-1)(f(x))^s$
\begin{eqnarray*}
F_j(x) &=& F_{j-1}(F(x))=F(x)+(j-1)\left(f(x+(f(x))^s)\right)^s\\
&=&  x+(f(x))^s+(j-1)(f(x))^s= x+j(f(x))^s.
\end{eqnarray*}
So, $F_p(x)=x$, completing the proof.
\end{proof}
In the case when $s$ is odd, the next corollary generalizes
 \cite[Theorem 4]{TuZeLiHe15} with a simple proof.
Notation is as in Proposition \ref{pr:var2}.
\begin{corollary}\label{cr:spe}
Let $p$ be an odd prime, $n=2k$ and $\delta \in \Sp\setminus\{0\}$. Then
\[
F(x)= L(x)+(x^{p^k}-x+\delta)^{\ell(p^k-1)+1},~1\leq\ell\leq p^k,
\]
permutes $\FF_{p^n}$ if and only if  $y\mapsto L(y)-2(-1)^\ell y$ permutes
$\Sp$. It is especially the case when:
\[
F(x)= \rho x+(x^{p^k}-x+\delta)^{\ell(p^k-1)+1},
  ~\rho\in\FF_{p^n}^*,~ \rho \neq 2(-1)^\ell.
\]
\end{corollary}
\begin{proof}
Since $p$ is odd, then $\ell(p^k-1)+1$ is odd for any $\ell$.
From Corollary \ref{cor:sdelta}, $F$ is a permutation if and only if 
\[
y\mapsto G(y)=L(y)-2(y+\delta)^s,~~s=\ell(p^k-1)+1
\]
is a permutation of $\Sp$. 
Note that  $\beta\in\Sp$  if and only if  $\beta^{p^k-1}= -1$. 
Moreover $\beta^s\in\Sp$ for any odd $s$, since
\[
T^{2k}_k(\beta^s)=(-\beta)^s+\beta^s=(-1)^s\beta^s+\beta^s=0.
\]
 For $y\in\Sp$, we have $y+\delta\in\Sp$ and
\[
(y+\delta)^s=(y+\delta)^{\ell(p^k-1)}(y+\delta)=(-1)^\ell(y+\delta).
\]
So, $G(y)=L(y)-2(-1)^\ell(y+\delta)$ and $G$ is a permutation if and only
 if the linear function   $y\mapsto L(y)-2(-1)^\ell y$ 
is bijective on $\Sp$. Now if $L(x)=\rho x$ then
 $y\mapsto (\rho-2(-1)^\ell) y$ is a permutation as soon as 
$\rho-2(-1)^\ell\ne 0$.
\end{proof}

\section{Conclusion}\label{sec:concl}
In this article several infinite classes of permutations have been specified.
 The existence of these specific classes of permutations  relies heavily  on 
the existence of linear translators.  To specify  $f:\FF_{p^n} \rightarrow \FF_{p^k}$ having linear translators, which
 are not 
monomials or binomials (or monomial trace forms), is left as an interesting research topic.  In Section \ref{se:spec}, we
 contribute to the  current works
on the functions of type (\ref{typ}).  We give another approach to analyze the permutation property by studying the balancedness of the
 component    functions, thus indicating
potentially another  research direction which would be the study of the spectrum
of the components of functions  of type (\ref{typ}).

\section{Acknowledgements}
Enes Pasalic is partly supported by the Slovenian Research Agency (research program P3-
0384 and research project J1-6720). Nastja Cepak is supported in part by the Slovenian Research Agency (research 25
program P3-0384 and Young Researchers Grant).

\end{document}